\documentclass[a4paper,oneside, 11pt]{article}

\usepackage[utf8]{inputenc}        
\usepackage[T1]{fontenc}             
\usepackage[english]{babel}
\usepackage{amsmath,amssymb,amsthm,bm}
\usepackage{microtype}
\usepackage{float}
\usepackage{a4wide}
\usepackage{tikz}
\usepackage{xspace}

\usepackage{enumitem}

\tikzstyle{dedge}=[->, thick]
\tikzstyle{lnode}=[scale=0.8]

\usepackage{breakurl}
\usepackage[breaklinks]{hyperref}

\usepackage[backend=bibtex,style=numeric]{biblatex}
\newtheorem{theorem}{Theorem}[section]
\newtheorem{corollary}[theorem]{Corollary}
\newtheorem{lemma}[theorem]{Lemma}

\theoremstyle{definition}
\newtheorem{definition}[theorem]{Definition}
\newtheorem*{implementation*}{Implementation}

\theoremstyle{remark}
\newtheorem*{remark*}{Remark}

\newcommand{\blockHeight}[1]{h\left( #1 \right)}
\newcommand{\prevote}{\textsc{Prevote}}
\newcommand{\approve}{\textsc{Approve}}
\newcommand{\precommit}{\textsc{Precommit}}
\newcommand{\lbft}{\textsc{Lisk-Bft}\xspace}
\newcommand{\monoRule}{\textsc{Monotonicity-Rule}\xspace}
\newcommand{\gst}{\text{GST}}
\newcommand{\longestchain}{\textsc{Longest}-\textsc{Chain}\xspace}

\newcommand{\validatorSet}{\mathcal{V}}
\newcommand{\validator}{V}

\begin{document}

\title{A lightweight BFT consensus protocol for blockchains}
\author{Jan Hackfeld \\ Lightcurve GmbH \\ \texttt{jan.hackfeld@lightcurve.io}}

\maketitle

\begin{abstract}
We present a general consensus framework that allows to easily introduce a customizable Byzantine fault tolerant consensus algorithm to an existing (Delegated) Proof-of-Stake blockchain. We prove the safety of the protocol under the assumption that less than~$1/3$ of the validators are Byzantine. The framework further allows for consensus participants to choose subjective decision thresholds in order to obtain safety even in the case of a larger proportion of Byzantine validators. Moreover, the liveness of the protocol is shown if less than~$1/3$ of the validators crash.

Based on the framework, we introduce \lbft, a Byzantine fault tolerant consensus algorithm for the Lisk ecosystem. \lbft integrates with the existing block proposal mechanism, requires only two additional integers in blocks and no additional messages. The protocol is simple and provides safety in the case of static validators if less than~$1/3$ of the validators are Byzantine. For the case of dynamically changing validators, we prove the safety of the protocol assuming a bound on the number of Byzantine validators and the number of honest validators that can change at one time.
We further show the liveness of the \lbft protocol for less than~$1/3$ crashing validators.
\end{abstract}

\section{Introduction}

The Paxos protocol introduced by Lamport in \cite{lamport98} is the basis for most consensus protocols solving the Byzantine agreement or more generally the state machine replication problem. Starting with~\cite{castro02}, there has been a multitude of different practical consensus protocols that adapt the basic Paxos protocol to improve various aspects such as the efficiency and the resistance against Byzantine faults. See \cite{rutti10} for a classification of different consensus protocols.

In this paper, we introduce \lbft, a Byzantine fault tolerant consensus algorithm for the Lisk blockchain that follows this line of research. We assume that there is a set of 
validators that can add blocks to a block tree via a given proposal mechanism (e.g. round-robin). 
Ideally, we would like the validators to propose one block after another always referencing the previously proposed block via a directed edge and hence forming a tree with only one growing branch, i.e., a blockchain. However, due to the latency of communication between the validators or deliberate attacks on the network, there can be multiple child blocks of the same parent block and separate growing branches. We therefore need a consensus protocol for validators to agree on one block at every height of the block tree.

We distinguish between \emph{honest} validators that correctly follow the consensus protocol
and \emph{Byzantine} validators that may behave arbitrarily. The consensus algorithm is
supposed to allow validators to \emph{decide} for exactly one block at every height of the block tree such that the following properties are satisfied:
\begin{itemize}
\item \textbf{Safety.} Two honest validators never decide for \emph{conflicting} blocks, i.e., blocks that are not contained in one branch of the block tree. 
\item \textbf{Liveness.} An honest validator eventually decides for a block at any height.
\item \textbf{Accountability.} A validator can detect if a Byzantine validator violates the consensus protocol and can identify the Byzantine validator.
\end{itemize}

Note that the safety property implies that all decided blocks for an honest validator are in one branch of the block tree and for any height two honest validators never decide for distinct blocks.

In the case of Bitcoin \cite{nakamoto08}, consensus is reached by validators choosing the chain which required the most ``work''. This rule does not satisfy both the safety and liveness property above because no matter how late a validator decides for a block~$B$, there is no guarantee that the block~$B$ will be contained in the Bitcoin blockchain the Bitcoin network agrees on in the future. The reason is that there always can be a different branch of the Bitcoin block tree in the future with more ``work'' although for economic reasons it becomes increasingly unlikely the more blocks are added to the block tree as descendants of~$B$.

Tendermint \cite{tendermint} is one of the first blockchain consensus protocols achieving the safety, liveness and accountability property above.  Namely, Tendermint guarantees that as long as~$<1/3$ of the validators are Byzantine, both the safety and liveness property hold, which is optimal for the partially synchronous system model introduced in~\cite{dwork88}. The Tendermint protocol proceeds in three distinct phases. Once the three phases are completed successfully, the validators decide for a block~$B$  and proceed with proposing and deciding for a descendant of~$B$. This way the block tree in Tendermint is simply a path as for every height the validators decide for a block before proceeding further. We call a consensus protocol with this property \emph{fork-free}. In contrast to that, we refer to a consensus protocol as \emph{forkful} if there is no requirement for validators to achieve consensus, i.e., decide for a block, before adding further blocks to the block tree. 

The \lbft protocol is a forkful protocol, where the block proposal can independently progress before consensus is reached on blocks. This idea of using \emph{speculative Byzantine fault tolerance} in order to decrease the overhead has already been considered in~\cite{kotla09} for the state machine replication problem. We also believe that a forkful protocol can achieve better performance in practice and significantly decrease the communication overhead as consensus
can be reached independent of the block proposal at any later point in time.
For a more detailed treatment of forkful versus fork-free consensus protocols and the message complexity involved, see \cite{buterin18}.

\section{Model and definitions}

In this section, we introduce the blockchain specific terminology as well as the underlying network model. We do not consider the general state machine replication problem, but instead assume that there is a block proposal mechanism, where time is divided into discrete slots and for every time slot there is a designated validator. We also change some of the state machine replication terminology to the blockchain specific context.

A \emph{block tree} is a directed tree with a designated root vertex referred to as \emph{genesis block} and a unique directed path from every vertex in the block tree to the genesis block. We refer to the vertices in the block tree as \emph{blocks}. We assume that blocks can contain arbitrary data such as signatures or messages. Every block~$B$ has a corresponding \emph{height}, denoted by~$h(B)$, which is the number of edges of the unique path from that block to the genesis block. Note that the genesis block has height 0. For a block~$B$ in the block tree, we call a block~$B'$ \emph{ancestor} if~$B'$ is distinct from~$B$ and is on the directed path from~$B$ to the genesis block. Moreover, $B'$ is called a \emph{descendant} of~$B$ if $B$ is an ancestor of~$B'$. If there is a directed edge from a block~$B$ to a block~$B'$, then $B'$ is called the \emph{parent} of~$B$ and $B$ a \emph{child} of~$B'$. We further refer to the subgraph induced by a path from a block without child to the genesis block as a \emph{branch} or a \emph{chain}. For a given chain, the unique block without child block is referred to as the \emph{tip} or \emph{head} of the chain. Two blocks are called \emph{conflicting} if they are not contained in one branch of the block tree, i.e., none is a descendant of the other.

We assume that there is a set of validators~$\validatorSet$ and a \emph{proposal mechanism} for validators to add child blocks of existing blocks to the block tree. We distinguish between three types of validators: 
A validator is \emph{honest} if it obeys the consensus protocol and actively participates by sending messages. 
A validator is further called \emph{offline} or \emph{crashed} if it stops participating in the consensus protocol from a certain point in time onwards, i.e., does not send any message more to the network, but all its messages obey the protocol rules. Moreover,  a \emph{Byzantine} validator can behave arbitrarily, e.g., maliciously violate the consensus protocol or crash.  Classic Byzantine fault tolerance consensus protocols typically yield guarantees assuming certain upper bounds on the number of Byzantine validators. In the case of blockchain, in particular Proof-of-Stake, we would rather have guarantees in terms of the stake represented by a certain set of validators. We can model this by associating a weight~$\omega_\validator \in [0,1]$ to every validator~$\validator$ such that $\sum_{\validator \in \validatorSet} \omega_\validator = 1$, where $\omega_\validator$ can be thought of as the relative stake in the case of a Proof-of-Stake blockchain. 
For the simplicity of exposition, we often do not use the validator weights explicitly but rather say that $>\omega$ (or $<\omega$)  of the validators satisfy a property~$X$ if the set~$M$ of validators satisfying property~$X$ has overall weight larger (or smaller) than~$\omega$, i.e., $\sum_{\validator \in M} \omega_\validator >\omega$ ($\sum_{\validator \in M} \omega_\validator <\omega$). For instance, we regularly assume that $<1/3$ of the validators are Byzantine. 

We further assume that the validators communicate by exchanging messages via the underlying network. We model the network using the \emph{partially synchronous system model} introduced in~\cite{dwork88}. The basic assumption of this model is that in general the network can behave arbitrarily badly, i.e., messages are lost or have a huge delay, but eventually, from a time~$\gst$ onward, the network behaves nicely, i.e., all messages arrive reliably and with a delay of at most~$\Delta$.
The parameter~$\gst$ is called the \emph{global stabilization time}. Messages are further always signed so that a validator can not be impersonated by another validator. The validators are not required to know any bound on~$\gst$ or~$\Delta$. 
Note that the eventually synchronous network is only necessary to show the liveness of the consensus algorithm, but not for the safety. In fact, this assumption on the underlying network is stronger than needed for our consensus to work in practice and our results can easily be extended to a network model with regular long enough good time periods, where the messages arrive reliably and within $\Delta$~time. These good periods only need to be long enough for validators to send enough messages to decide for blocks. For the ease of exposition, however, we use the partially synchronous system model.

Note that due to the impossibility results of Fischer et al. \cite{fischer85,fischer86}, stating that consensus is impossible without randomization in an asynchronous message-passing system with at least one Byzantine or crashing node, the partially synchronous system model is one of the weakest models where consensus without randomization is possible. 
Known consensus algorithms for asynchronous systems using randomization have worse guarantees in terms of maximum number of Byzantine validators, require a large number of local computations or a larger number of messages \cite{cachin05, king14}.

\section{General consensus algorithm framework}
In this section, we define a general consensus protocol with two types of messages, which
is later used as the basis for the \lbft protocol.

\subsection{Protocol rules}
We now describe the protocol rules of the general consensus framework. A validator~$\validator$ can send two types of messages, 
$\prevote(B,T,\validator)$ and $\precommit(B,T,\validator)$, where $B$ and $T$ are blocks in the block tree such that $T=B$ or $T$ is a descendant of~$B$ ($T$ is typically the tip of the current chain of~$\validator$). In an implementation, $B$ and $T$ can be block hashes and $\validator$ the public key and signature of the validator~$\validator$. 
We assume that both types of messages are added to the blocks so that we can refer to the \emph{prevotes} and \emph{precommits} included in the blockchain from the genesis block up to the tip~$T$. 
Moreover, we shortly write that a validator prevotes for~$B$ if $\validator$ sends a~$\prevote(B,T,\validator)$ message and equivalently that $\validator$ precommits for~$B$ if it sends a message~$\precommit(B,T,\validator)$ for a block~$T$.
Next, we describe the protocol rules that must be satisfied when sending prevote and precommit messages and the conditions for validators to decide for blocks.

\begin{definition}\label{def-consensus-framework}
The prevotes and precommits send by a validator~$\validator \in \validatorSet$ must obey the following rules:
\begin{enumerate}[label=(\Roman*)]
\item For any two distinct messages~$\prevote(B,T,\validator)$ and $\prevote(B',T',\validator)$, it holds that $h(B)\neq h(B')$. 
\label{rule_1}
\item For any message~$\precommit(B,T,\validator)$, the blockchain up to block~$T$ must include both \label{rule_2}
\begin{enumerate}[label = (\alph*), ref=\theenumi{}~(\alph*)]
\item one $\prevote(B,\cdot,\validator)$ message and
\item $\prevote(B,\cdot,\cdot)$ messages by~$>2/3$ of the validators.
\end{enumerate}
\item For any two messages  $\precommit(B,T,\validator)$ and  $\prevote(B',T',\validator)$ with $h(B)<h(B')$ at least one of the following conditions is satisfied: \label{rule_3}
\begin{enumerate}[label = (\alph*), ref=\theenumi{}~(\alph*)]
\item $B'$ is a descendant of~$B$. \label{rule_3_a}
\item The blockchain up to~$T'$ contains a block~$\bar{B}$ at height~$h(\bar{B})\geq h(B)$ and $\prevote(\bar{B},\cdot,\cdot)$ messages by~$>2/3$ of the validators.\label{rule_3_b}
\end{enumerate}
\end{enumerate}
Let $\tau_\validator \in (\tfrac{1}{3},1]$ denote the \emph{decision threshold} of a validator~$\validator$. Once validator~$\validator$ receives $\precommit(B,\cdot, \cdot)$ messages by~$>\tau_\validator$ of the validators, validator~$\validator$ \emph{decides for} block~$B$ or \emph{finalizes} block~$B$. This implies that $\validator$ also decides for all ancestors of~$B$.
\end{definition}

The basic intuition behind the protocol rules is that prevotes must not contradict, so validators can cast only one prevote for every height by~\ref{rule_1}. Moreover, in order to send a precommit for a block, a validator has to have send a prevote for that block and received prevotes  by~$>2/3$ of the validators for it due to~\ref{rule_2}. Finally, the intuition behind~\ref{rule_3} is that a validator cannot change to a different chain without a justification, i.e., after a precommit for block~$B$, all prevotes of larger height must be for descendants of~$B$ or, as a justification, there is a block~$\bar{B}$ at height~$h(\bar{B})\geq h(B)$ with $>2/3$~prevotes that allows the validator to start sending prevotes on a different chain.

In the next subsection, we show how the decision threshold relates to the safety property of the protocol. Intuitively, the higher the decision threshold of the validators, the more validators can be Byzantine without losing the safety property. On the other hand, for the liveness property, the decision threshold cannot be too high as otherwise a validator may never decide on a block. 

\subsection{Safety}

In this section, we show that the safety property is fulfilled for the general consensus protocol assuming a bound on the overall weight of Byzantine validators that depends on the decision thresholds of the validators.

\begin{theorem}[Safety]\label{theo-safety}
If the weight of Byzantine validators is $<(\tau -1/3)$ for~$\tau \in (1/3, 1]$, then any two honest validators with decision threshold at least~$\tau$ never decide for conflicting blocks. 
\end{theorem}

\begin{proof}
Assume, for the sake of contradiction, that there are two conflicting blocks~$B_1$ and~$B_2$ 
as well as an honest validator~$\validator_1$ with decision threshold~$\tau_1$ deciding for~$B_1$ and an honest validator~$\validator_2$ with decision threshold~$\tau_2$ deciding for~$B_2$. Without loss of generality let $h(B_1) \leq h(B_2)$. Note that by assumption we have $\tau\leq \tau_1$ and $\tau \leq \tau_2$.

If $\validator_1$ decides for~$B_1$, it must have received a $\precommit(B_1,\cdot, \cdot)$ messages by~$>\tau_1$ of the validators by definition. 
By assumption, $\tau_1 \geq \tau$ and $<(\tau -1/3)$ of the validators are Byzantine so that there is a set of~$>1/3$ honest or crashing validators~$\validatorSet_1$ that must have sent a $\precommit(B_1,\cdot, \cdot)$ message. 
Note that crashing validators obey the protocol rules, but may stop sending messages from some point onwards.
In particular, by property~\ref{rule_1}, no validator~$\validator \in \validatorSet_1$ sends a $\prevote(B,\cdot, \validator)$ message for a distinct block~$B$ at height~$h(B_1)$ and hence no distinct block~$B$ at height~$h(B_1)$ can receive $>2/3$~prevotes.
This also implies that we cannot have $h(B_1)=h(B_2)$ as there can only be $\leq 2/3$~prevotes for~$B_2$ and hence no validator~$\validator$ obeying the protocol rules would send a $\precommit(B_2,\cdot,\validator)$ message by~\ref{rule_2}. Thus, the threshold of~$\tau_2$~precommit messages for block~$B_2$ cannot be reached and we must have~$h(B_1)<h(B_2)$.

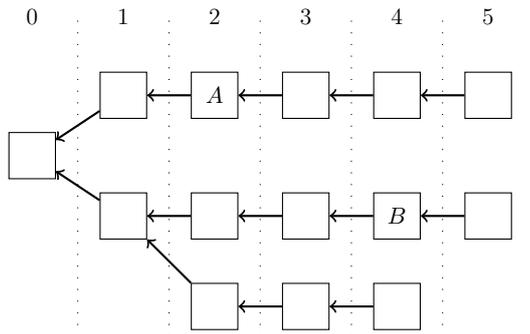
\begin{figure}[H]
\centering
\tikzstyle{block}=[rectangle,draw, minimum width=2em,minimum height=2em,align=center,scale=0.8]
\tikzstyle{dedge}=[<-, thick]

\newcommand{\x}{1.2}
\newcommand{\y}{0.8}
\begin{tikzpicture}
\node[block] (b1) at (0,0){};

\node[block] (b2) at (1*\x,\y) {};
\node[block] (b3) at (2*\x,\y) {$B_1$};
\node[block] (b4) at (3*\x,\y) {};
\node[block] (b5) at (4*\x,\y) {};
\node[block] (b6) at (5*\x,\y) {};

\node[block] (c1) at (1*\x,-\y) {};
\node[block] (c2) at (2*\x,-\y) {};
\node[block] (c3) at (3*\x,-\y) {};
\node[block] (c4) at (4*\x,-\y) {$B_2$};
\node[block] (c5) at (5*\x,-\y) {};

\node[block] (d2) at (2*\x,-2.5*\y) {};
\node[block] (d3) at (3*\x,-2.5*\y) {};
\node[block] (d4) at (4*\x,-2.5*\y) {};

\foreach \i in {1,...,5} {
	 \pgfmathtruncatemacro\j{\i+1}
	\draw[dedge](b\i)--(b\j);

}
\draw[dedge](b1)--(c1);
\foreach \i in {1,...,4} {
	 \pgfmathtruncatemacro\j{\i+1}
	\draw[dedge](c\i)--(c\j);

}
\draw[dedge](c1)--(d2);
\draw[dedge](d2)--(d3);
\draw[dedge](d3)--(d4);

\foreach \i in {0,...,5} {
\node[scale=0.8] (l1) at  (\i*\x,2.3*\y){\i};
\draw[loosely dotted] (\i*\x+0.5*\x,-3*\y) -- (\i*\x+0.5*\x,2.3*\y);
}
\end{tikzpicture}
\caption{Block tree with two conflicting blocks~$B_1$ and~$B_2$.}\label{fig-blocktree-conflicting-blocks}
\end{figure}

By~\ref{rule_3}, every validator~$\validator\in  \validatorSet_1$ only sends a prevote message for a block that is not a descendant of~$B_1$ if there is a block~$\bar{B}$ at height~$h(\bar{B})\geq h(B_1)$ with $\prevote(\bar{B},\cdot,\cdot)$ messages by~$>2/3$ of the validators. 
As all validators in~$\validatorSet_1$ follow the protocol, there have to be $>2/3$~prevotes for such a block~$\bar{B}$ by validators not in~$\validatorSet_1$. 
This is a contradiction, as $\validatorSet_1$ is a set of $>1/3$~validators by assumption. 
Hence, no block that is not a descendant of~$B_1$ can obtain $>2/3$~prevotes and thus no validator following the protocol sends a precommit for~$B_2$. 
This means that the decision threshold of~$\tau_2$ precommit messages for block~$B_2$ cannot be reached contradicting that the honest validator~$\validator_2$ decided for~$B_2$.
\end{proof}

\begin{remark*}
Note that for $\geq 1/3$~Byzantine validators a consensus protocol that guarantees both safety and liveness  in the partially synchronous system model is impossible as shown in~\cite{dwork88}. This is also the case for the consensus framework above because Byzantine validators can for example not send any prevotes  and thus prevent any decision for a block as the threshold of $>2/3$~prevotes is not reached. In order to have both the liveness and safety property, we hence cannot have $\geq 1/3$~Byzantine validators.
The purpose of the flexible subjective decision threshold is to allow for optimistic participants (those assuming a small fraction of Byzantine validators) to decide for blocks earlier by only requiring a small decision threshold~$\tau_\validator$. On the other hand, pessimistic participants can decide for a block only with a very high decision threshold (with the risk of never making a decision if too many validators are Byzantine or offline) with the guarantee that $\geq \tau_\validator -1/3$ must violate the protocol rules for a conflicting block to be decided. In particular, in the case that slashing conditions are added to punish Byzantine validators by burning their deposits, this allows to choose a decision threshold
depending on the desired safety guarantee in terms of deposit at stake.
\end{remark*}

\subsection{Liveness}

For proving the liveness condition of the consensus framework from Definition~\ref{def-consensus-framework}, we rely on the partially synchronous system model. Recall that the basic assumption is that after the global stabilization time of~$\gst$, all messages arrive reliably with a delay of at most~$\Delta$. We assume that the time slots for validators to add blocks to the block tree have a duration of at least~$2 \Delta$. This way the underlying network is fast enough for the message about a new block to reach every validator and all validators to send new prevote and precommit messages that reach every other validator before the next block is proposed. This way, all prevotes and precommits can be included in the next block and validators at regular times have the same information about the block tree (we assume that blocks, prevote and precommit messages that were lost before time~$\gst$ are simply rebroadcast).

Additionally, for the liveness property to hold, validators need to follow the same \emph{fork choice rule}, i.e., a rule that defines how the validators choose one branch of the block tree as their current chain. More precisely, a fork choice rule is a function~$f$ that, given the block tree~$\mathcal{T}_\validator$ of validator~$\validator$ (i.e., the block tree containing all blocks received by validator~$\validator$) and the arrival times of all blocks as input, returns one branch of the block tree~$\mathcal{T}_\validator$. We refer to this branch as the \emph{canonical chain} of~$\mathcal{T}_\validator$. We say that a validator~$\validator$ \emph{follows a fork choice rule}~$f$ if it adds new proposed blocks to the canonical chain and only sends prevote and precommit messages for blocks in the canonical chain. For the liveness proof, we use the fork choice rule defined below. 

\begin{definition}[\longestchain]
Let~$\mathcal{T}_\validator$ be the current block tree of validator~$\validator$ and $\bar{h}$ be the largest height of a block in~$\mathcal{T}_\validator$ with prevotes by~$>2/3$ of the validators included in one branch of~$\mathcal{T}_\validator$.  
Then the \longestchain fork choice rule returns the longest branch that contains a block~$\bar{B}$ at height~$\bar{h}$ and $\prevote(\bar{B},\cdot,\cdot)$ messages by~$>2/3$ of the validators. Ties are broken in favor of the complete chain received first.
\end{definition}

The essential property that the fork choice rule needs to satisfy is that the branch returned always contains a block~$\bar{B}$ of largest height with prevote messages by~$>2/3$ of the validators. The reason is that by~\ref{rule_3}, once a validator casts a precommit for~$\bar{B}$, it can only prevote and precommit for blocks on the same branch as~$\bar{B}$ (unless there is a block of height at least~$h(\bar{B})$ with $>2/3$~prevotes). Thus, by requiring the canonical chain to include~$\bar{B}$, the validators automatically satisfy~\ref{rule_3} when casting prevotes and precommits. In general, other fork choice rules that satisfy that the canonical chain always contains~$\bar{B}$ are possible and yield the liveness property. An example is \emph{Immediate message-driven GHOST} proposed for Ethereum \cite{buterin18forkchoice}.

If~$<1/3$ of the validators are Byzantine, there can be at most one block at a given height with $>2/3$~prevotes by the the validators by~\ref{rule_1}. Thus, in this case $\bar{B}$ in the definition above is unique.
Further, note that the tie-breaking in favor of the chain seen first is rather arbitrary and
also block hashes, i.e., identifiers associated with every block, can be used. The main point of the tie-breaking is to have a rule that quickly makes the block tree change to a state (by the new proposed blocks) where there are no more ties.

In the following theorem, we show liveness assuming a bound on the weight of crashed validators. The reason for not considering
Byzantine validators is that this would require to make stronger assumptions on the block proposal mechanism. Furthermore, for a concrete application of the general consensus framework it is usually necessary to show liveness again for the concrete block proposal mechanism and additional rules how validators send prevote and precommit messages.

\begin{theorem}[Liveness]\label{theo-general-consensus-liveness}
Assume all validators obey the general consensus protocol and \longestchain fork choice rule and  ~$<1/3$ of the validators crash. Then for any~$l \in \mathbb{N}$, an honest validator~$\validator$ with decision threshold~$\tau_\validator \in (1/3,2/3]$ will eventually decide on a block at height~$l$.
\end{theorem}

\begin{proof}
Consider an arbitrary set of blocks proposed and prevotes as well as precommits cast by all validators according to the consensus protocol from Definition~\ref{def-consensus-framework}. 
Assume that we have reached time~$\gst+2\Delta$.
This means that all messages arrive  reliably within time~$\Delta$ and all previously lost messages have been broadcast again and received by all validators.
In particular, the block tree~$\mathcal{T}_\validator$ at the beginning of every block proposal time slot is the same for every 
validator~$\validator$. Note that due to lost or delayed messages before time~$\gst$, there can initially be multiple chains of the same length
that contain a block of largest height with prevotes by~$>2/3$ of the validators. Because of different arrival
times of blocks, validators could decide for different canonical chains. However, after the first block is 
proposed, these ties are broken and the \longestchain rule yields the same canonical chain for every honest validator. 

Let~$l\in \mathbb{N}$ and $h_{\text{max}}$ be the largest height that any honest validator cast a prevote for. 
We show that an honest validator~$\validator$ with decision threshold~$\tau_\validator \in (1/3,2/3]$ will eventually decide on a block at height~$l$.
The canonical chain of all honest validators continues growing until there are two consecutive blocks~$B_1$
and~$B_2$ added to the chain with~$h(B_i)>\max\{l, h_{\text{max}}\}$ for~$i=1,2$.

We first show that any honest validator~$\validator$ can send a $\prevote(B_1,\cdot, \validator)$ message.
Assume~$\validator$ previously cast a precommit message for a block~$B'$. We show that the precommit for~$B'$ does not prevent $\validator$ to prevote for~$B_1$. By assumption, $B_1$ was added to the chain that contains the block of largest height~$\bar{B}$ for which~$>2/3$ of the validators cast prevotes (at the time when $B_1$ is added to the chain).
If $B'$ is on a branch not containing~$\bar{B}$, then by \ref{rule_3} validator~$\validator$ can prevote for~$B_1$ because we have $h(B')< h(\bar{B})$ as $B'$ must have $>2/3$~prevotes and $\overline{B}$ is the block of largest height with $>2/3$~prevotes. Note that there can be at most one block at every height with $>2/3$~prevotes because of~\ref{rule_1} and the fact that $>2/3$ of the validators honestly follow the protocol.
If $B'$ is on the same branch as~$\overline{B}$, then by the choice of~$\overline{B}$ either $B'=\overline{B}$ or $B'$ is an an ancestor of~$\overline{B}$. In both cases, the precommit for~$B'$ allows $\validator$ to prevote for any descendants of~$B'$, in particular, for~$B_1$. Thus, $\validator$ can prevote for~$B_1$.

Therefore, all~$>2/3$ honest validators can send $\prevote(B_1,\cdot,\cdot)$ messages and these messages reach all validators before the next block is proposed because we assume that the block proposal time slots have duration~$>2 \Delta$. This means that the honest validator of~$B_2$ can include these prevote messages for~$B_1$ into the block~$B_2$. After receiving~$B_2$,
all honest validators can now send a $\precommit(B_1,\cdot,\cdot)$ messages. Hence, $B_1$ receives precommits by~$>2/3$ of the validators and any validator with decision threshold~$\tau \in (1/3,2/3]$ will decide for~$B_1$ and all ancestors of~$B_1$.
In particular, any honest validator will decide for a block at height~$l$ since $h(B_1)>l$. This shows the claim.
\end{proof}

\subsection{Accountability}

We assume that all messages $\prevote(B,T,\validator)$ and $\precommit(B,T,\validator)$ are signed by the validator~$\validator$.
Using the current chain up to block~$T$, any validator can verify if the messages sent by a validator~$\validator$ satisfy the consensus protocol rules~\ref{rule_1}--\ref{rule_3}.
Note that we assume that it is impossible that a validator sends a message exclusively to one honest validator without all other validators eventually being able to learn about it, as the underlying peer-to-peer network will gossip the messages to all validators.

\subsection{Lightweight consensus messages}

In this section, we describe how multiple prevote and precommit messages can be concisely expressed via a single message, which we call an \emph{approve}  message. An $\approve(k,p,T,\validator)$ message contains a block~$T$, two integers~$k,p \in \mathbb{N}$ with~$k < h(T)$ and $p < h(T)$ as well as the validator~$\validator$ sending the message. In an implementation, $T$ can be the hash of the block~$T$ and $\validator$ the public key and signature of the message by validator~$\validator$.
Moreover, for the message~$\approve(k,p,T,\validator)$ to be valid, $p$ must be the height of the block~$\bar{B}$ of largest height that is an ancestor of~$T$ and the chain up to~$T$ includes prevotes
by~$>2/3$ of the validators for~$\bar{B}$. As a convention, we assume that the genesis block has prevotes by all validators. The height~$p$ can be viewed as redundant information that can be derived from the chain up to block~$T$. However, it is essential for a simple expression of the protocol rules and short proofs of any violation of the protocol rules.
We now state the rules that validators must satisfy when sending approve messages.

\begin{definition}[\monoRule]\label{def-approve-rules}
The approve messages by a validator~$\validator$ satisfy the \monoRule if for any two distinct approve messages $\approve(k,p,T,\validator)$ and $\approve(k',p',T',\validator)$ with~$k\leq k'$ the following conditions are satisfied:
\begin{enumerate}[label=(\roman*)]
\item It holds that $h(T)\leq k'$. \label{condition_A_approval}
\item We have $p\leq p'$.\label{condition_B_approval}
\end{enumerate}
\end{definition}

Let $B_0,B_1,\ldots, B_{l}$ be the current chain of validator~$\validator$ and $T=B_{r}$ for~$r \in \{0,1,\ldots, l\}$. Then the intuition behind an $\approve(k,p,T,\validator)$ message is that $\validator$ approves the blocks $B_{k+1},B_{k+1},\ldots, B_r$.
We require that $k < h(T)$ holds, such that a non-empty set of blocks up to block~$T$ is approved. 
Now, condition~\ref{condition_A_approval} of the \monoRule requires two  approve messages to be for blocks with disjoint sets of heights.
Note that, in particular, we cannot have $k=k'$ for two distinct messages as otherwise we would obtain $h(T)>k=k'$.
Moreover,  \ref{condition_B_approval} ensures that a validator only switches to another chain extending a different block with $>2/3$~prevotes if that block has no smaller height. 

We now describe how an approve message translates to prevote and precommit messages.  

\begin{definition}\label{def-trans-approve}
Let $B_0,B_1,\ldots, B_{l}$ be one branch of the block tree of validator~$\validator$. 
Then an $\approve( k, p, B_r, \validator)$ message with $r\in \{k+1,\ldots, l\}$ implies the following prevotes and precommits: 
\begin{enumerate}[label=(\alph*)]
\item $\validator$ casts a $\prevote(A, B_r,\validator)$ for every block~$A \in \{B_{k+1},B_{k+2}, \ldots, B_r\}$. \label{transform_approval_A}
\item Considering the prevotes and precommits included in the chain~$C=B_0,B_1,\ldots, B_r$ (given via  approve  messages), let 
\begin{align*}
j_1 & :=\max  \left( \{0 \leq s \leq r \mid \exists \ \precommit(B_s,\cdot, \validator) \text{ in chain } C \} \cup \{-1\} \right) \\
j_2 &:=\max  \left( \{0 \leq  s \leq k \mid \nexists \ \   \prevote(B_s,\cdot, \validator) \text{ in chain } C\} \cup \{-1\} \right) \\
j & :=\max\ \{j_1, j_2\}+1.
\end{align*}
Then $\validator$ casts a $\precommit(A,B_r,\validator)$ for every block~$A$ in the set~$\{B_j,\ldots, B_r\}$ that has (implied) prevotes by~$>2/3$ of the validators included in the chain~$C$. 
\label{transform_approval_B}
\end{enumerate}
\end{definition}

The idea is that the validator approves the blocks $B_{k+1},B_{k+2}, \ldots, B_r$ by prevoting for them. Moreover, validator~$\validator$ sends a maximum number of precommits for blocks starting from its previous precommit while satisfying protocol rule \ref{rule_2}. Note that for the precommits in \ref{transform_approval_B}, it is important that only the prevotes in the chain~$C$ and not those cast in \ref{transform_approval_A} are taken into account, so that for any $\precommit(A,B_r,\validator)$ we have~$h(A)\leq p$. This property is needed in the proof of Lemma~\ref{lem-approve-satisfies-protocol-rules}.

Next, we show that if a validator satisfies the \monoRule for the approve messages, then the implied prevote and precommit messages satisfy the protocol rules from Definition~\ref{def-consensus-framework}.  

\begin{lemma}\label{lem-approve-satisfies-protocol-rules}
If the approve messages by a validator~$\validator$ satisfies the \monoRule and $\validator$ sends no additional prevote or precommit messages, then the prevotes and precommits implied by the transformation in Definition~\ref{def-trans-approve} satisfy the protocol rules~\ref{rule_1}--\ref{rule_3}.
\end{lemma}

\begin{proof}
Consider two messages \approve$(k,p,T,\validator)$ and \approve$(k',p',T',\validator)$ with~$k\leq k'$. Because of condition~\ref{condition_A_approval} of the \monoRule, we have $h(T) \leq k'$ and thus
\begin{align*}
\{k+1,k+2, \ldots, h(T)\} \cap \{k'+1,k'+2, \ldots, h(T')\} = \emptyset.
\end{align*}
Therefore all prevotes are cast for blocks at distinct heights, i.e., \ref{rule_1} is satisfied. 

Next, consider an arbitrary $\precommit(A,B_r,\validator)$ that is cast according to the transformation of an $\approve(k,p,B_r,\validator)$ message, where we use the notation from Definition~\ref{def-trans-approve}. 
We then have $A\in \{B_j,B_{j+1},\ldots, B_r\}$, where $j=\max \{j_1,j_2\}+1$. As $j_2< j$, validator~$\validator$ must have cast a prevote for~$A$ and by definition of the transformation there are $\prevote(A,\cdot, \cdot)$ messages by~$>2/3$ of the validators included in the chain up to~$B_r$.
Hence, \ref{rule_2} is satisfied for all precommit messages.

Finally, consider a $\precommit(A,T,\validator)$ message and a $\prevote(A',T',\validator)$ message with~$h(A)<h(A')$. If~$A'$ is a descendant of~$A$, then \ref{rule_3_a} is satisfied. Otherwise, we must show that \ref{rule_3_b} is satisfied, i.e.,  there exists a
block~$\bar{B}$ at height~$h(\bar{B})\geq h(A)$ with $\prevote(\bar{B},\cdot, \cdot)$ messages by~$>2/3$ of the validators
contained in the blockchain up to~$T'$.

From now on, assume that \ref{rule_3_a}  does not hold. Then both $\precommit(A,T,\validator)$ and $\prevote(A',T',\validator)$  cannot be implied by the same message $\approve(k,p,T,\validator)$, as otherwise $A'$ is a descendant of~$A$. Thus, there have to be two distinct messages $\approve(k,p,T,\validator)$ and $\approve(k',p',T',\validator)$ such that the first implies the message $\precommit(A,T,\validator)$ and the second implies the message  $\prevote(A',T',\validator)$. As $\prevote(A',T',\validator)$ is implied by $\approve(k',p',T',\validator)$, we must have $k'<h(A')\leq h(T')$. 
For the sake of contradiction, let us assume that $k\geq k'$. Then the branch up to~$T$ cannot contain any implied prevote by~$\validator$ for a block at height~$h(A')$, because we assume that this branch does not contain~$A'$ (otherwise again $A'$ is a descendant of~$A$), $\validator$ already send a  $\prevote(A',T',\validator)$ message on a different branch and \ref{rule_1} is satisfied. 
Using the notation from Definition~\ref{def-trans-approve}, this implies that $j_2\geq h(A')$ for the chain~$C$ up to~$T$. 
Therefore any precommit implied by the message $\approve(k,p,T,\validator)$ must be for blocks at height larger than~$h(A')$. This would mean $h(A)>h(A')$, a contradiction. Hence, we must have~$k<k'$. 

As $k<k'$ holds, we must have $p\leq p'$ by condition~\ref{condition_B_approval} of the \monoRule.
Note that we must have $h(A)\leq p$ as $p$ is the largest height of a block with prevotes by~$>2/3$ of the validators included in the chain up to~$T$. Let $\bar{B}$ be the block at height~$p'$ that is an ancestor of~$T'$. We have $h(A)\leq p\leq p'=h(\bar{B})$ and there are  $\prevote(\bar{B},\cdot, \cdot)$ messages by~$>2/3$ of the validators contained in the blockchain up to~$T'$. Thus,~\ref{rule_3} is also satisfied. 
\end{proof}

By Lemma~\ref{lem-approve-satisfies-protocol-rules}, the protocol rules~\ref{rule_1}--\ref{rule_3} are satisfied when validators cast prevotes and precommits implicitly via approve messages. This means that also Theorem~\ref{theo-safety} holds for the case that
validators only use approve messages for the consensus algorithm.

\begin{corollary}
Assume for the consensus algorithm validators only use approve messages that are transformed according to Definition~\ref{def-trans-approve} and for any honest or crashing validator the approve messages satisfy the \monoRule.
If the weight of Byzantine validators is $<(\tau -1/3)$ for~$\tau \in (1/3, 1]$, then any two honest validators with decision threshold at least~$\tau$ never decide for conflicting blocks. 
\end{corollary}

\subsection{Dynamic set of validators}

So far, we only considered a static set of validators~$\validatorSet$. In blockchains networks, however, this set is typically changing over time as validators may want to stop participating in Proof-of-Stake and withdraw their deposit or different delegates get voted in the top ranks in a Delegated Proof-of-Stake system. 
The interesting question is what the implications of a  dynamically changing set of validators are for the safety and liveness guarantees from the last sections. 

In this section, we consider the approach that a change of validators (e.g., because of withdrawals, deposits or votes) is always possible and comes into effect without requiring finality, i.e., without requiring that the block or a descendant of the block containing the information of the change of validators receives precommits by a certain threshold of validators. 
Note that if the set of  validators changes over time, there can now be blocks at the same height which assume a different set of validators. 
For instance, one chain may contain the information regarding a change of validators, whereas another chain does not. 
Thus, we now associate a unique set of \emph{active validators}~$\validatorSet_B$ to every block~$B$ of the block tree, which are the active consensus participant for that block. 
For one branch $B_0,B_1,B_2\ldots$ of the block tree, the active validators for a block~$B_l$ in that branch must be given by the information in the ancestor blocks in that branch, i.e., blocks $B_0,B_1,\ldots,B_{l-1}$. 
In particular, the genesis block~$B_0$ defines $\validatorSet_{B_1}$, the active validators for block~$B_1$. 
The active validators  associated to a block~$B$ always contain the validator forging~$B$ and every validator that can cast prevotes and precommits for the block~$B$. 
We further say that a validator~$\validator$ is active at block~$B$ if $\validator \in \validatorSet_B$. 

Furthermore, apart from a change in validators, it can happen that the weight~$\omega_\validator$ associated with a validator~$\validator$ changes (e.g., due to a change in stake). For the following results, we need to distinguish between the changed and unchanged weights of validators. Hence, we model an increase of weight of validator~$\validator$ from~$\omega_\validator$ to~$\omega_{\validator}'$ by assuming that a new validator~$\validator'$ with weight~$\omega_{\validator}'-\omega_\validator$ is added to the new set of validators. Similarly, for a decrease of weight from~$\omega_\validator$ to~$\omega_{\validator}'$, we assume that the original set of validators contains two validators~$\validator$ and $\validator'$ with weights~$\omega_{\validator'}$ and $\omega_\validator-\omega_{\validator'}$, respectively. The change of weight in this case is then modeled by~$\validator'$ leaving the set of validators. This allows us to only consider a change in the set of validators and not additionally consider changes of weights. Note that for any active set of validators~$\validatorSet_B$ for a block~$B$ the weights always need to satisfy $\sum_{\validator \in \validatorSet_B} \omega_\validator=1$.

Moreover, if validators cast prevotes and precommits implicitly via approve messages, then for a $\approve( k, p, B_r, \validator)$  message to be valid, we require~$\validator$ to be active at the blocks $B_{k+1},\ldots, B_r$. 
This way any implied $\prevote(A, B_r,\validator)$ cast according to Definition~\ref{def-trans-approve} is for a block~$A$ where $\validator$ is active. As a precommit is only implied if $\validator$ previously made an implicit prevote, it holds that also all implied precommits are for blocks where $\validator$ is active. 

We now show sufficient conditions for the safety and liveness property for the case of validators that change according to the mechanism described above.

\begin{theorem}\label{theo-unconstrained-change-validators}
Let $\mathcal{B}$ be the set of blocks of the block tree. Then the following properties hold:
\begin{enumerate}[label=(\alph*)]
\item If $\bigcap_{B\in \mathcal{B}} \validatorSet_B$ contains $>(4/3-\tau)$ validators following the consensus protocol, then two honest validator in~$\bigcup_{B\in \mathcal{B}} \validatorSet_B$ with decision threshold at least~$\tau \in (1/3,1]$ never decide for conflicting blocks (Safety).
\item 
Assume all active validators in the block tree follow the general consensus protocol and \longestchain fork choice rule and $<1/3$ of the validators in any active validator set crash. 
If after the global stabilization time is reached, there is regularly an unchanged validator set for two consecutive blocks, then for any $l \in \mathbb{N}$, an honest validator~$\validator$ with decision threshold~$\tau_\validator \in (1/3,2/3]$ will eventually decide on a block at height~$l$ (Liveness).
\end{enumerate}
\end{theorem}

\begin{proof}
The proof of the first part of the claim is similar to the proof of Theorem~\ref{theo-safety}.
Assume, for the sake of contradiction, that there are two conflicting blocks~$B_1$ and~$B_2$ 
as well as an honest validator~$\validator_1 \in \bigcup_{B\in \mathcal{B}} \validatorSet_B$ with decision threshold~$\tau_1$ deciding for~$B_1$ and an honest validators~$\validator_2\in \bigcup_{B\in \mathcal{B}} \validatorSet_B$ with decision threshold~$\tau_2$ deciding for~$B_2$. 
By assumption, we have $\tau\leq \tau_1$ and $\tau \leq \tau_2$. Hence, there must be $\precommit(B_1,\cdot, \cdot)$ 
messages cast by~$>\tau$ of the validators in~$\validatorSet_{B_1}$ and also $\precommit(B_2,\cdot, \cdot)$ 
messages cast by~$>\tau$ of the validators in~$\validatorSet_{B_2}$.

Let $\validatorSet$ be the set of $>(4/3-\tau)$~validators obeying the protocol that is contained in $\bigcap_{B\in \mathcal{B}} \validatorSet_B$ by assumption. 
Note that $\validatorSet$ can contain both honest as well as crashing validators.
In particular, we have $\validatorSet \subseteq \validatorSet_{B_1}\cap \validatorSet_{B_2}$.
The set of validators $\validatorSet_{B_1} \setminus \validatorSet$ has weight $<1-(4/3-\tau)= \tau -1/3$, but there are  $\precommit(B_1,\cdot, \cdot)$ 
messages cast by~$>\tau$ of the validators in~$\validatorSet_{B_1}$. Hence, there are $>1/3$~validators in~$\validatorSet$ that sent a $\precommit(B_1,\cdot, \cdot)$ message. Analogously,
there are also $>1/3$~validators in~$\validatorSet$ that must have sent a $\precommit(B_2,\cdot, \cdot)$ message. 
 
Without loss of generality, we assume $h(B_1) \leq h(B_2)$.  
Let $\validatorSet_1 \subseteq \validatorSet$ be the subset 
of $>1/3$~validators that send a precommit for~$B_1$. By~\ref{rule_2}, every validator~$\validator \in \validatorSet_1$ 
must have also cast a  $\prevote(B_1,\cdot, \validator)$ message. Thus, by~\ref{rule_1}, no validator in~$\validator \in \validatorSet_1$ 
sends a prevote for any other block at height~$h(B_1)$. We must therefore have $h(B_1)<h(B_2)$ as otherwise $B_2$ cannot obtain $>2/3$~prevotes
because $\validatorSet_1$ is contained in every active set of validators.

By~\ref{rule_3}, every validator~$\validator\in \validatorSet_1$ only sends a prevote message for a block that is not a descendant of~$B_1$ if there is a block~$\bar{B}$ at height~$h(\bar{B})\geq h(B_1)$ with $\prevote(\bar{B},\cdot,\cdot)$ messages by~$>2/3$ of the validators. As all validators in~$\validatorSet_1$ obey the protocol, there have to be $>2/3$~prevotes for such a block~$\bar{B}$ by validators not in~$\validatorSet_1$. This is a contradiction, as $\validatorSet_1$ is a set of $>1/3$~validators
that is contained in any active set of validators. Hence, no block that is not a descendant of~$B_1$ can obtain $>2/3$~prevotes and thus no validator obeying the protocol sends a precommit for~$B_2$. This means that the decision threshold of $\tau_2$~precommit messages for block~$B_2$ cannot be reached contradicting that the honest validator~$\validator_2$ decided for~$B_2$.

The liveness property follows from Theorem~\ref{theo-general-consensus-liveness}, as the proof only requires an unchanged set of honest validators with weight~$>2/3$ for two consecutive blocks.
\end{proof}

The above theorem shows that the safety guarantee is only slightly weaker if only a small fraction of the validators change. However, several changes of a small fraction of validators over time could amount to almost all validators changing. The following result shows that as long as the validators finalize blocks in between these changes of validators, we can improve the safety guarantee above.

\begin{theorem}\label{theo-constrained-change-validators}
Let $\alpha \in (0,1/3)$ and $\tau \in (1/3+2 \alpha, 1]$. 
Let $B, B'$ be two blocks in one branch such that $B'$ is the descendant of~$B$ with minimal height such that the chain up to~$B'$ contains $>\tau$~precommits for~$B$ or a descendant of~$B$. 
We assume that between all such pairs of blocks~$B$ and $B'$ in the block tree $<\alpha$~of the validators obeying the protocol change. 
If any active validator set contains $>4/3-\tau+2 \alpha$~validators obeying the protocol, then two honest validators with decision threshold at least~$\tau$ never decide for conflicting blocks.
\end{theorem}

\begin{proof}
Let us assume, for the sake of contradiction, that there are two conflicting blocks~$B_1$ and $B_2$ such that an honest validator~$\validator_1$ with decision threshold~$\tau_1$ decides for~$B_1$ and another honest validators~$\validator_2$ with decision threshold~$\tau_2$ decides for~$B_2$. 
Let $A_0$ be the block of largest height that is a common ancestor of~$B_1$ and $B_2$. 
Moreover, let~$A_1$ be the first descendant of~$A_0$ on the branch containing~$B_1$ such that the branch contains $>\tau$~precommits for~$A_1$. 
Similarly, let $A_2$ be the first descendant of~$A_0$ on the branch containing~$B_2$ such that the branch contains $>\tau$~precommits for~$A_2$. 
The setting is shown in Figure~\ref{fig-blocktree-conflicting-blocks-changing-validators}. As $\tau_1\geq \tau$, there must be at least~$\tau$ precommits for~$B_1$ contained in the branch containing~$B_1$ and therefore we must have $A_1=B_1$ or $A_1$ is an ancestor of~$B_1$. 
The same holds for~$A_2$ and $B_2$.
Without loss of generality let~$h(A_1) \leq h(A_2)$. 

Let $\validatorSet_0$ be the subset of all validators in~$\validatorSet_{A_0}$ that obey the protocol. 
The weight of all validators in~$\validatorSet_0$ is $>4/3-\tau+2 \alpha$ by assumption.
We also know by assumption that $<\alpha$~of the validators obeying the protocol change before a block on a branch receives $>\tau$~precommit. 
This means that the weight of validators in $\validatorSet_0\setminus \validatorSet_{A_i}$ is $<\alpha$ for $i \in \{1,2\}$.
Hence, $\validatorSet:=\validatorSet_0 \cap \validatorSet_{A_1}\cap \validatorSet_{A_2}$ is a set of validators of weight $>4/3 -\tau$ that obey the protocol. 
The validator set $\validatorSet$ is further contained in the active set of validators associated to any block on the path from~$A_0$ to~$A_1$ or the path from~$A_0$ to~$A_2$. 

The set of validators $\validatorSet_{A_1} \setminus \validatorSet$ has weight $<1-(4/3-\tau)= \tau -1/3$, but there are $\precommit(A_1,\cdot,\cdot)$ by~$>\tau$ of the validators in the chain containing~$B_1$.
Hence, there must be a subset~$\validatorSet_1\subseteq \validatorSet$ of $>1/3$~validators that cast a precommit for~$A_1$. 
This also implies that we cannot have $h(A_1)=h(A_2)$ as there can only be $\leq 2/3$~prevotes for~$A_2$ and hence no validator~$\validator\in  \validatorSet_{A_2}$ obeying the protocol would cast a $\precommit(A_2,\cdot,\validator)$ by~\ref{rule_2}. Thus, the threshold of $\tau$~precommits for block~$A_2$ cannot be reached in this case and we must have $h(A_1)<h(A_2)$.

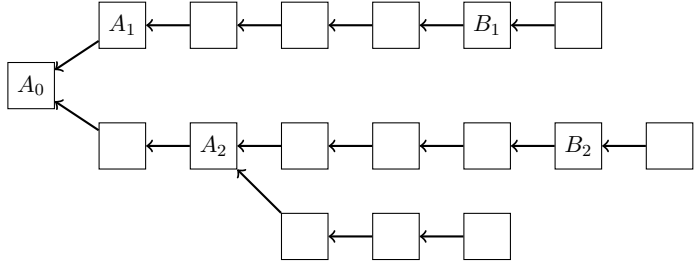
\begin{figure}[H]
\centering
\tikzstyle{block}=[rectangle,draw, minimum width=2em,minimum height=2em,align=center,scale=0.8]
\tikzstyle{dedge}=[<-, thick]

\newcommand{\x}{1.2}
\newcommand{\y}{0.8}
\begin{tikzpicture}
\node[block] (b1) at (0,0){$A_0$};

\node[block] (b2) at (1*\x,\y) {$A_1$};
\node[block] (b3) at (2*\x,\y) {};
\node[block] (b4) at (3*\x,\y) {};
\node[block] (b5) at (4*\x,\y) {};
\node[block] (b6) at (5*\x,\y) {$B_1$};
\node[block] (b7) at (6*\x,\y) {};

\node[block] (c1) at (1*\x,-\y) {};
\node[block] (c2) at (2*\x,-\y) {$A_2$};
\node[block] (c3) at (3*\x,-\y) {};
\node[block] (c4) at (4*\x,-\y) {};
\node[block] (c5) at (5*\x,-\y) {};
\node[block] (c6) at (6*\x,-\y) {$B_2$};
\node[block] (c7) at (7*\x,-\y) {};

\node[block] (d2) at (3*\x,-2.5*\y) {};
\node[block] (d3) at (4*\x,-2.5*\y) {};
\node[block] (d4) at (5*\x,-2.5*\y) {};

\foreach \i in {1,...,6} {
	 \pgfmathtruncatemacro\j{\i+1}
	\draw[dedge](b\i)--(b\j);

}
\draw[dedge](b1)--(c1);
\foreach \i in {1,...,6} {
	 \pgfmathtruncatemacro\j{\i+1}
	\draw[dedge](c\i)--(c\j);

}
\draw[dedge](c2)--(d2);
\draw[dedge](d2)--(d3);
\draw[dedge](d3)--(d4);

\end{tikzpicture}
\caption{Block tree with two conflicting blocks~$B_1$ and~$B_2$.}\label{fig-blocktree-conflicting-blocks-changing-validators}
\end{figure}

By~\ref{rule_3}, every validator~$\validator\in  \validatorSet_1$ only sends a prevote message for a block that is not a descendant of~$A_1$ if there is a block~$\bar{B}$ at height~$h(\bar{B})\geq h(A_1)$ with $\prevote(\bar{B},\cdot,\cdot)$ messages by~$>2/3$ of the validators. 
As all validators in~$\validatorSet_1$ obey the protocol, there have to be $>2/3$~prevotes for such a block~$\bar{B}$ by validators not in~$\validatorSet_1$. 
This is a contradiction, as $\validatorSet_1$ is a set of $>1/3$~validators that is contained in the active set of validators associated with every block on the path from~$A_0$ to~$A_2$.  
Hence, block~$A_2$ cannot obtain $>2/3$~prevotes and thus no validators obeying the protocol sends a precommit for~$A_2$. 
This means that the threshold of $\tau$~precommits for block~$A_2$ cannot be reached, which is a contradiction.
\end{proof}

\begin{remark*}
Theorem~\ref{theo-constrained-change-validators} outlines a way how we can maintain close to optimal safety guarantees while allowing dynamic validators. Assume that we only allow a change of validators with overall weight~$<1/100$ at a time, i.e., after a change of validators from~$\validatorSet$ to~$\validatorSet'$ we require the new set of validators~$\validatorSet'$ to finalize a block before another change of validators can occur. If we assume that we always have $>2/3+2/100$~validators following the protocol at any given time, then Theorem~\ref{theo-constrained-change-validators} yields that the safety property is satisfied. 

Allowing an arbitrary large change of validators at a given time while preserving the safety property under the assumption that at any time $>2/3$ of the obey the protocol, requires a much more complex mechanism of two set of validators, i.e., both $\validatorSet$ and $\validatorSet'$, to vote on the same blocks.
For the \emph{Casper the Friendly Finality Gadget} consensus algorithm this approach is sketched in~\cite{buterin17}.
\end{remark*}

\section{Weighted \lbft consensus protocol}

We now describe \lbft, a consensus protocol for the Lisk ecosystem,  which is based on the general consensus framework from the last section. 
The consensus protocol is flexible enough to work with different configurations of the Delegated Proof-of-Stake validator selection mechanism in the Lisk ecosystem. 

In the Delegated Proof-of-Stake system in Lisk Mainnet, in every round of 103 consecutive blocks there are 103 distinct delegates that can propose a block. 
The 103 delegates include the 101 delegates with largest vote weight and two randomly selected additional delegates. 
Only the 101 delegates with largest vote weight are supposed to participate in the consensus protocol and cast prevotes and precommits, whereas the two additional delegates only propose blocks. 
This is modeled by setting the active set of validators to the 103 forging delegates for each block of the round, assigning uniform weights of ~$\omega_\validator=1/101$ to each of the 101 delegates with largest vote weight and weight $\omega_\validator=0$ to the two additional forging delegates. .
In particular, we then have $\sum_{\validator \in \validatorSet_B} \omega_\validator = 1$ for every block~$B$ with with active validator set~$\validatorSet_B$. 
For different numbers of forging delegates per round and delegates participating in the  \lbft consensus protocol the weights can be assigned analogously such that  $\sum_{\validator \in \validatorSet_B} \omega_\validator = 1$ holds for every block.

Instead of using separate consensus messages, in \lbft all prevotes and precommits are implied by two additional integers added to blocks using
the same idea as the lightweight consensus messages introduced in the last section. 
Hence, the additional two integers in the blocks are the only overhead of the \lbft consensus protocol.

\begin{definition}[\lbft]\label{def-lisk-bft}
Let~$\tau \in (1/3, 1]$ and $\rho\in \mathbb{N}$ be two protocol parameters. 
Let $C:=B_0,B_1,\ldots, B_{l-1}$ be the current chain of validator~$\validator$ and $h_0\in \{0,1,\ldots, l\}$ be the height of the first block in the current chain since when $\validator$ has been continuously in the active validator set.
\begin{enumerate}
\item When $\validator$ proposes block~$B_{l}$ as a child of~$B_{l-1}$, it adds the following information to the block: 
\begin{description}[style=multiline,leftmargin=2cm,font=\normalfont]
\item[$h_{\text{previous}}$] The largest height of a block that $\validator$ previously proposed (on any branch) and $0$ if $\validator$ has not proposed any block.
\item[$h_{\text{prevoted}}$] The height of the block of largest height with $>2/3$~prevotes in the current chain $B_0,B_1,\ldots, B_{l-1}$.
\end{description}
In general, for a block~$B$ we use the notation~$h_{\text{previous}}(B)$ and $h_{\text{prevoted}}(B)$ to refer to the values above.
\item 
If $h_{\text{previous}}(B_l)<h(B_l)$ holds, then the information in block~$B_{l}$ implies prevotes and precommits for blocks in the following order: 
\begin{enumerate}[label=(\alph*),ref=\theenumi{}~(\alph*)]
\item Let $k:=\max \{h_{\text{previous}}(B_l), h_0-1, l - \rho\}$. 
Then the block implies a $\prevote(A, B_l,\validator)$ cast for every block~$A \in \{B_{k+1},B_{k+2}, \ldots, B_l\}$. \label{lisk-bft-transformation-a}
\item Considering the implied prevotes and precommits included in the chain~$C$ (via the additional information in the block), let 
\label{lisk-bft-transformation-b}
\begin{align*}
j_1 & :=\max \left(  \{0 \leq s < l \mid \exists \ \precommit(B_s,\cdot, \validator) \text{ in chain } C \} \cup \{-1\} \right) \\
j_2 &:=\max \left(  \{0 \leq  s \leq h_{\text{previous}}(B_l) \mid \nexists \ \   \prevote(B_s,\cdot, \validator) \text{  in chain } C\} \cup \{-1\} \right) \\
j & :=\max\ \{j_1, j_2, h_0-1, l - \rho\}+1.
\end{align*}
Then the block implies a $\precommit(A,B_l,\validator)$ for every block~$A$ in the set~$\{B_j,\ldots, B_l\}$ that has prevotes by~$>2/3$~validators included in the chain~$C$. 
\label{transform_liskbft_B}
\end{enumerate}
\item A block~$B$ and all ancestors of~$B$ are finalized if there are precommits for~$B$ by $>\tau$~validators included in the chain~$C$.
\item Validators follow the \longestchain
rule, where the block of largest height with~$>2/3$ prevotes in a chain $A_0,A_1,\ldots, A_s$ is computed taking only the prevotes implied by blocks $A_0,A_1,\ldots, A_{s-1}$, and not those implied by~$A_s$, into account. 
\item We assume that the block proposal mechanisms assign block slots to the active validators in a round-robin fashion. 
That is, if a validator set~$\validatorSet$ is first active for a block~$B$, then the block slot for block~$B$ and the next $|\validatorSet|-1$ blocks slots are each assigned to one distinct validator in $\validatorSet$. 
We further require that $\rho \geq 3 \cdot |\validatorSet_B|$ holds for any block~$B$ in the block tree.
\end{enumerate}
\end{definition} 

The \lbft protocol is an application of the general consensus framework from the last section,
where validators only cast prevotes and precommits via the two additional integers contained in blocks.
The values~$h_{\text{previous}}$ and~$h_{\text{prevoted}}$ in block~$B_l$ correspond to the validator~$\validator$ first proposing block~$B_l$ and afterwards sending the message  $\approve(h_{\text{previous}},h_{\text{prevoted}},B_l,\validator)$. 
However, the transformation of the two integers in blocks to prevote and precommit messages above contains two modifications compared to the transformation for approve messages in Definition~\ref{def-trans-approve}. 
First of all, we have $k=\max \{h_{\text{previous}}, h_0-1, l - \rho\}$ instead of~$k=h_{\text{previous}}$. 
The reason for enforcing $k\geq h_0-1$ is to ensure that the implied prevotes are only for blocks at heights where validator~$\validator$ was active. 
We further require that $k\geq  l - \rho$ holds for performance reasons so that it is sufficient to consider the previous $\rho-1$~blocks in the chain when updating the number of prevotes for blocks. 
Note that we require that $\rho \geq 3 \cdot |\validatorSet_B|$ holds for any block~$B$ in the block tree so that~$\rho$ is large enough such that a block can be finalized after $\rho-1$~additional blocks are added on top, i.e., the prevotes and precommits implied by~$\rho-1$ subsequent blocks are sufficient such that a block can reach~$>\tau$ precommits, see the proof of the liveness condition in Theorem~\ref{theo-lisk-bft-props} for details.
Secondly, we have $j =\max\ \{j_1, j_2, h_0-1, l - \rho\}+1$ instead of $j :=\max\ \{j_1, j_2\}+1$. Again, this ensures that all implied precommits are for blocks where $\validator$ was active and also for performance reasons all precommits are only for a subset of the previous $\rho-1$~blocks.
Overall, it is important to note that the prevotes and precommits implied by the transformation above are a subset of the prevotes and precommits implied
according to Definition~\ref{def-trans-approve} when instead sending the message $\approve(h_{\text{previous}},h_{\text{prevoted}},B_l,\validator)$.

A subtlety of the transformation~\ref{lisk-bft-transformation-b} above is that we consider only the prevotes implied by all blocks up to~$B_{l-1}$ and not the prevotes implied by~\ref{lisk-bft-transformation-a}. This is analogous to the transformation for approve messages in Definition~\ref{def-trans-approve}.

We establish some basic properties of the \lbft protocol in the following lemma. 

\begin{lemma}\label{lem-lisk-bft}
Let $B_1$ and $B_2$ be two distinct blocks proposed by a validator~$\validator$ obeying the \lbft protocol. Then the following properties hold: 
\begin{enumerate}[label=(\alph*)]
\item The \longestchain rule implies that a branch with tip~$B$ is chosen over a branch with tip~$A$, if the following conditions are satisfied: \label{lem-lisk-bft-property-a}
\begin{align*}
h_{\text{prevoted}}(B)> h_{\text{prevoted}}(A) \quad \text{ or } \quad  
\left(h_{\text{prevoted}}(B)= h_{\text{prevoted}}(A) \text{ and } h(B)>h(A)\right).
\end{align*}
\item If $B_1$ is proposed before~$B_2$, then the following inequalities hold:\label{lem-lisk-bft-property-b}
\begin{align*}
h_{\text{previous}}(B_1) \leq h_{\text{previous}}(B_2), \quad 
h_{\text{prevoted}}(B_1) \leq h_{\text{prevoted}}(B_2), \quad 
h(B_1) \leq h_{\text{previous}}(B_2).
\end{align*}
\item If $B_1$ is proposed before~$B_2$ and $h_{\text{prevoted}}(B_1)=h_{\text{prevoted}}(B_2)$ holds, then we must have~$h(B_1)<h(B_2)$.\label{lem-lisk-bft-property-c}
\item Block~$B_1$ is proposed before~$B_2$ if and only if \label{lem-lisk-bft-property-d}
\begin{align}
(h_{\text{previous}}(B_1),h_{\text{prevoted}}(B_1),h(B_1))<_{\text{lex}} (h_{\text{previous}}(B_2),h_{\text{prevoted}}(B_2),h(B_2)), \label{inequ-lexicographic}
\end{align}
where $<_{\text{lex}}$ is the lexicographical ordering.
\item If any two blocks by validator~$\validator$ satisfy the inequalities in \ref{lem-lisk-bft-property-b}, then the implied prevotes and precommits according to the \lbft protocol satisfy the conditions of the general consensus framework in Definition~\ref{def-consensus-framework}. \label{lem-lisk-bft-property-e}
\end{enumerate}
\end{lemma}

\begin{proof}
\begin{enumerate}[label=(\alph*)]
\item According to the \longestchain rule, the first deciding factor is the largest height of a block with $>2/3$~prevotes contained in some branch of the block tree. For a chain with tip~$A$, the largest height of a block with $>2/3$~prevotes is given by~$h_{\text{prevoted}}(A)$  because the prevotes implied by block~$A$ are not taken into account. If there are multiple chains that contain a block with $>2/3$~prevotes of the same height, then the longer chain containing these prevotes has priority according to the \longestchain rule. This immediately yields the condition in the claim.

\item As $B_2$ is proposed after~$B_1$ and the largest height of a block previously proposed is non-decreasing, we must have $h_{\text{previous}}(B_1)\leq h_{\text{previous}}(B_2)$. As validator~$\validator$ follows the  \longestchain rule, the tips of the canonical chain of~$\validator$ satisfy that $h_{\text{prevoted}}(\cdot)$ is non-decreasing by claim  \ref{lem-lisk-bft-property-a}. Thus, we must have $h_{\text{prevoted}}(B_1)\leq h_{\text{prevoted}}(B_2)$.
Moreover, when proposing $B_2$, the largest height of a block previously proposed must be at least~$h(B_1)$. This implies $h(B_1)\leq h_{\text{previous}}(B_2)$.

\item By claim \ref{lem-lisk-bft-property-a}, every time validator~$\validator$ switches to a different chain, either $h_{\text{prevoted}}(\cdot)$ is strictly increasing or it remains unchanged and the height of the tip strictly increases.
The same holds every time a new child is added to the tip of the chain of validator~$\validator$ (either proposed by~$\validator$ or another validator).
Thus, if $h_{\text{prevoted}}(B_1)=h_{\text{prevoted}}(B_2)$ holds, then we must have $h(B_1)<h(B_2)$.

\item Assume $B_1$ is proposed before~$B_2$. By claim \ref{lem-lisk-bft-property-b}, we must have $h_{\text{previous}}(B_1) \leq h_{\text{previous}}(B_2)$ and $ h_{\text{prevoted}}(B_1) \leq h_{\text{prevoted}}(B_2)$. Moreover, by \ref{lem-lisk-bft-property-c}, we have $h(B_1)<h(B_2)$ if $h_{\text{prevoted}}(B_1) = h_{\text{prevoted}}(B_2)$. Thus, Inequality~\eqref{inequ-lexicographic} must hold.
The other direction follows by contraposition.

\item We need to show that the implied prevotes and precommits according to the transformation in Definition~\ref{def-lisk-bft} satisfy the protocol rules
\ref{rule_1}--\ref{rule_3} of the general consensus framework. The proof works analogous to the proof of Lemma~\ref{lem-approve-satisfies-protocol-rules}.

Let $T$ and $T'$ be two distinct blocks proposed by a validator~$\validator$, where $T'$ is proposed after~$T$. By claim~\ref{lem-lisk-bft-property-b}, we have $h(T)\leq h_{\text{previous}}(T')$. Any $\prevote(A,T,\validator)$ implied by~$T$ satisfies
$h(A)\leq h(T)$ and any  $\prevote(A',T',\validator)$ implied by~$T'$ satisfies $h(A')> h_{\text{previous}}(T')\geq h(T)$. This means that all implied prevotes are for distinct heights, i.e., \ref{rule_1} is satisfied. 

Next, consider some $\precommit(A,B_l,\validator)$ implied by a block~$B_l$, where we use the notation from Definition~\ref{def-lisk-bft}. We then have $h(A)> j_2$
and hence $\validator$ must have cast a prevote for~$A$ (implied by~\ref{lisk-bft-transformation-a} or a previous block by~$\validator$ in the current chain). By the definition of the transformation there are $>2/3$ prevotes for~$A$ in the chain up to~$B_{l-1}$.
Hence, \ref{rule_2} is satisfied for all precommit messages.

Finally, for showing \ref{rule_3}, consider a $\precommit(A,T,\validator)$ and a $\prevote(A',T',\validator)$  with $h(A)<h(A')$ implied by two blocks~$T$ and $T'$. If $A'$ is a descendant of~$A$, then \ref{rule_3} is satisfied. Otherwise, we must show that there exists a
block~$\bar{B}$ at height~$h(\bar{B})\geq h(A)$ with $\prevote(\bar{B},\cdot, \cdot)$ by~$>2/3$ of the validators contained in the chain up to block~$T'$.

We first show that $T'$ must be proposed after~$T$. Assume, for the sake of contradiction, that $T'$ is proposed before~$T$. By claim~\ref{lem-lisk-bft-property-b}, we then have $h(T')\leq h_{\text{previous}}(T)< h(T)$ and, 
in particular, $h(A')\leq h(T') < h(T)$.
By our assumption, $A'$ is not a descendant of~$A$ so
$A'$ cannot be on the branch containing~$A$ and $T$. But then the branch containing~$A$ and $T$ does not contain any prevote for a block at height~$h(A')$ because \ref{rule_1} is satisfied. This means that any precommit implied by~$T$ must be for blocks at height at least~$h(A')+1$, which contradicts $h(A)<h(A')$ and that $\precommit(A,T,\validator)$ is implied by~$T$.
Hence, $T'$ must be proposed after~$T$. 

By claim~\ref{lem-lisk-bft-property-b}, we then obtain $h_{\text{prevoted}}(T) \leq h_{\text{prevoted}}(T')$. We further have $h(A) \leq h_{\text{prevoted}}(T)$ because $
h_{\text{prevoted}}(T)$ is the largest height of a block with prevotes by~$>2/3$ of the validators included in the chain up to~$T$. Let $\bar{B}$ be the block at height~$h_{\text{prevoted}}(T')$ that is an ancestor of~$T'$. We have 
$h(A)\leq h_{\text{prevoted}}(T) \leq h_{\text{prevoted}}(T')=h(\bar{B})$
and there are $\prevote(\bar{B},\cdot, \cdot)$ by~$>2/3$~validators contained in the blockchain up to~$T'$. Thus,~\ref{rule_3} is also satisfied. \qedhere
\end{enumerate}
\end{proof}

Let $A$ and $B$ be two distinct blocks proposed by a validator~$\validator$. Using Lemma~\ref{lem-lisk-bft}~\ref{lem-lisk-bft-property-d}, we can derive which of the blocks was proposed first assuming that validator~$\validator$ follows the \lbft protocol. We call the blocks~$A$ and $B$ \emph{contradicting} if, assuming this proposal order, they do \emph{not} satisfy one of the inequalities in Lemma~\ref{lem-lisk-bft}~\ref{lem-lisk-bft-property-b} or the condition of Lemma~\ref{lem-lisk-bft}~\ref{lem-lisk-bft-property-c}.

Lemma~\ref{lem-lisk-bft} directly yields the following two properties: 
The blocks proposed by a validator obeying  the \lbft protocol are never contradicting. Secondly, if none of the blocks of a validator are contradicting, then the implied prevotes and precommits satisfy the conditions of the general consensus framework in Definition~\ref{def-consensus-framework} due to Lemma~\ref{lem-lisk-bft}~\ref{lem-lisk-bft-property-e}. Hence, in order to detect whether the implied prevotes and precommits of a validator violate the general consensus framework, it is sufficient to check whether any pair of blocks by that validator is contradicting.

Note that we obtain all properties in the above paragraph if we only require blocks to satisfy all inequalities in Lemma~\ref{lem-lisk-bft}~\ref{lem-lisk-bft-property-b}. A violation of Lemma~\ref{lem-lisk-bft}~\ref{lem-lisk-bft-property-c} only implies a violation of the fork choice rule, but not that the prevotes and precommits violate the general consensus framework. For the liveness property, it is, however, important that validators choose one chain according to the fork choice rule and, in particular, do not forge multiple blocks with the same height and the same $h_{\text{prevoted}}$ on different chains because this causes a tie in the fork choice rule if the blocks are at the tip of the chains.

The next lemma allows to reduce the number of checks to ensure that the blocks of a validator are not contradicting. Instead of checking the inequalities in Lemma~\ref{lem-lisk-bft}~\ref{lem-lisk-bft-property-b} and the condition of Lemma~\ref{lem-lisk-bft}~\ref{lem-lisk-bft-property-c} for any pair of blocks by the same validator, it is sufficient to only check pairs of successive blocks by the same validator in one chain.

\begin{lemma}
Let $C := B_0, B_1, \ldots, B_{l}$ be a chain of blocks and $B_i, B_j, B_k$ blocks forged by validator~$\validator$ such that $0 \leq i <j <k\leq l$. If the pairs of blocks $ (B_i,B_j)$ and $(B_j,B_k)$ are not contradicting, then the pair of blocks $(B_i, B_k)$ is also not contradicting.
\end{lemma}

\begin{proof}
By assumption, we have $h_{\text{previous}}(B_i) \leq h_{\text{previous}}(B_j)$ and $h_{\text{previous}}(B_j) \leq h_{\text{previous}}(B_k)$. 
We therefore obtain $h_{\text{previous}}(B_i) \leq h_{\text{previous}}(B_k)$. The same argument shows  $h_{\text{prevoted}}(B_i) \leq h_{\text{prevoted}}(B_k)$. We further have $h(B_i)  \leq h_{\text{previous}}(B_j)\leq h_{\text{previous}}(B_k)$ and therefore $(B_i,B_k)$ satisfies all three inequalities in Lemma~\ref{lem-lisk-bft}~\ref{lem-lisk-bft-property-b}. 

Moreover, if $h_{\text{prevoted}}(B_i) = h_{\text{prevoted}}(B_k)$ holds, then we must have both $h_{\text{prevoted}}(B_i) = h_{\text{prevoted}}(B_j)$ and
$h_{\text{prevoted}}(B_j) = h_{\text{prevoted}}(B_k)$.
This implies $h(B_i)<h(B_j)$ and $h(B_j)<h(B_k)$ by assumption. Hence, we also have $h(B_i)<h(B_k)$ implying that $B_i$ and $B_k$ are not contradicting.
\end{proof}

Using Lemma~\ref{lem-lisk-bft}, we now show the desired safety, liveness and accountability properties of the \lbft protocol.
Note that for showing the liveness property we assume that~$<1/3$~of the validators are offline, i.e., not forging any blocks, and we do not consider Byzantine validators. As validators cast prevotes and precommits together with the block proposal in the \lbft protocol, Byzantine validators could intentionally create competing branches in the block tree so that the prevotes and precommits of some honest validators are not considered as their block is not eventually included in the canonical chain. Such an attack would only be successful if a substantial number of Byzantine validators coordinate the creation of competing branches and at regular times multiple Byzantine validators can consecutive propose blocks so that the competing chain becomes longer than the chain previously forged by the honest validators. We believe that such an attack is rather unlikely as it requires a lot of coordination and a favorable validator ordering for the attackers. Further, it only effects the liveness property, but not the safety property.

\begin{theorem}\label{theo-lisk-bft-props}
The \lbft protocol with parameters ~$\tau \in (1/3, 1]$ and $\rho\in \mathbb{N}$ has the following properties:
\begin{enumerate}[label=(\alph*)]
\item 
Let $\mathcal{B}$ be the set of blocks of the block tree.
If $\bigcap_{B\in \mathcal{B}} \validatorSet_B$ contains $>(4/3-\tau)$~obeying the \lbft protocol, then two honest validator in~$\bigcup_{B\in \mathcal{B}} \validatorSet_B$ with decision threshold at least~$\tau$ never finalize conflicting blocks.
\label{lisk-bft-safety1}
\item 
Assume that for any branch in the block tree $<\alpha$~among the validators obeying the \lbft protocol change without a block obtaining $>\tau$~precommits. 
If any active validator set contains $>4/3-\tau+2 \alpha$~validators obeying the protocol, then two honest validators with decision threshold at least~$\tau$ never decide for conflicting blocks.
\label{lisk-bft-safety2}
\item 
Assume all active validators in the block tree follow the general consensus protocol and \longestchain fork choice rule and $<1/3$ of the validators in any active validator set crash. 
If after the global stabilization time is reached, there is regularly an unchanged validator set for $\rho$ consecutive blocks, then for any $l \in \mathbb{N}$, an honest validator~$\validator$ with decision threshold~$\tau_\validator \in (1/3,2/3]$ will eventually decide on a block at height~$l$.\label{lisk-bft-liveness}
\item If a Byzantine validator forges blocks such that the implied prevotes and precommits violate the general consensus framework, then this protocol violation can be detected and the Byzantine validator can be identified. 
\label{lisk-bft-accountability}
\end{enumerate}
\end{theorem}

\begin{proof}
By Lemma~\ref{lem-lisk-bft}, the implied prevotes and precommits in the \lbft protocol satisfy the \lbft protocol rules \ref{rule_1}--\ref{rule_3}. Then the first claim follows directly from Theorem~\ref{theo-unconstrained-change-validators} and the second claim from Theorem~\ref{theo-constrained-change-validators}.

For the liveness property in the third claim, consider an arbitrary state of the block tree and an arbitrary set of blocks proposed according
to the \lbft protocol. 
Moreover, we assume that we have reached the global stabilization time~$\gst$ so that all messages arrive reliably within time~$\Delta$. 
Let further $l \in \mathbb{N}$ and $h_{\text{max}}$ be the largest height that any validator proposed a block for.
We show that an honest validator with decision threshold~$\tau_\validator \in (1/3,2/3]$ eventually finalizes a block at height~$l$.

After time~$\gst+\Delta$, all honest validators are aware of all competing chains. There could be multiple competing chains for which $h_{\text{previous}}(\cdot)$ and $h(\cdot)$ of the tip of the chain is the same, but one block proposed by an honest validator always breaks these ties. Afterwards, in every proposal slot, there is one unique block~$A$ such that for every other block~$A'$ either $h_{\text{prevoted}}(A)> h_{\text{prevoted}}(A')$ or $h_{\text{prevoted}}(A)= h_{\text{prevoted}}(A')$ and $h(A)>h(A')$ holds. This means $A$ is the tip of the canonical chain of every honest validator. 

This canonical chain for all honest validators will continue growing until its tip has height $>\max\{l, h_{\text{max}}\}$. 
By assumption, there will be~$\rho$ consecutive blocks~$B^{(1)}, \ldots,B^{(\rho)}$ appended to the canonical chain such that $\blockHeight{B^{(1)}}>\max\{l, h_{\text{max}}\}$ and $\bigcap_{i \in \{1,\ldots,\rho\}} \validatorSet_{B^{(i)}}$ contains an honest set of validators~$\validatorSet_H$ of weight $>2/3$. 
We show that $B^{(1)}$ will be considered final by any honest validator $\validator$ with decision threshold $\tau_\validator \in (1/3,2/3]$.

The  \lbft protocol assumes that the block proposal mechanism assigns block slots in a round-robin fashion and that  $3 |\validatorSet_{B^{(i)}} | \leq \rho$ for $i \in \{1,\ldots,\rho\}$. 
In particular, this implies that there is $j \in  \{1,\ldots,\rho\}$ such that every validator in 
$\validatorSet_H$ forges at least one of the blocks in $B^{(1)}, \ldots,B^{(j)}$ and at least one of the blocks in $B^{(j+1)}, \ldots,B^{(\rho)}$. 
The first block $B^{(i)}$ in $B^{(1)}, \ldots,B^{(j)}$ proposed by an honest validator~$\validator\in\validatorSet_H $ satisfies $h_{\text{previous}}\left(B^{(i)}\right)<\blockHeight{B^{(1)}}$ as we assume there is only one growing canonical chain and hence $\validator$ did not forge any block at height $\geq h(B^{(1)})$ already. 
By the transformation in Definition~\ref{def-lisk-bft}, $B^{(i)}$ implies a prevote for~$B^{(1)}$ as~$\validator$ was active at height~$\blockHeight{B^{(1)}}$ and $\blockHeight{B^{(i)}}-\blockHeight{B^{(1)}}\leq \rho -1$. 
Hence, after $B^{(j)}$ is added to the chain, every validator in $\validatorSet_H$ has cast a prevote for $B^{(1)}$, i.e., $B^{(1)}$ has prevotes by $>2/3$ of the validators included in the chain up to~$B^{(j)}$.
If a validator $\validator\in \validatorSet_H$ has not already cast a precommit for~$B^{(1)}$, then the first subsequent block~$B^{(i)}$ with~$i \geq j+1$ by~$\validator$ implies a precommit by~$\validator$ for~$B^{(1)}$, because the following holds for the transformation in Definition~\ref{def-lisk-bft}:
\begin{itemize}[itemsep=2pt]
\item $j_2<\blockHeight{B^{(1)}}$, as all honest validators add blocks to the same canonical chain extending~$B^{(1)}$,
\item $j_1<\blockHeight{B^{(1)}}$, because otherwise $\validator$ already made a precommit for~$B^{(1)}$,
\item $h_0 \leq \blockHeight{B^{(1)}}$, since all validators in~$\validatorSet_H$ are continuously active, 
\item $\blockHeight{B^{(i)}}-\blockHeight{B^{(1)}}\leq \rho -1$.    
\end{itemize}
Hence, after $B^{(\rho)}$ is added to the chain, there are precommits for~$B^{(1)}$ by every validator in $\validatorSet_H$ in the current chain as each of them forges at least one block in $B^{(j+1)}, \ldots,B^{(\rho)}$ by assumption. 
As the weight of validators in  $\validatorSet_H$ is larger than $2/3$, the chain up to block~$B^{(\rho)}$ therefore contains precommits by~$>2/3$ precommits for~$B^{(1)}$. Thus, any honest validator $\validator$ with decision threshold $\tau_\validator \in (1/3,2/3]$ can decide for $B^{(1)}$ and all ancestors of $B^{(1)}$. 
In particular, any honest validator finalizes a block at height~$l$ as~$\blockHeight{B^{(1)}}>l$.

For the last claim, note that Lemma~\ref{lem-lisk-bft}
\ref{lem-lisk-bft-property-b} states the conditions that must hold if a validator follows the \lbft protocol. By Lemma~\ref{lem-lisk-bft}~\ref{lem-lisk-bft-property-e}, these are also
sufficient for the implied prevotes and precommits to satisfy \ref{rule_1}--\ref{rule_3} of the general consensus framework. In order to determine whether the implied prevotes and precommits of blocks forged by a validator~$\validator$ obey the protocol rules~\ref{rule_1}--\ref{rule_3}, it is therefore sufficient to consider every pair of blocks~$B_1$ and $B_2$ proposed by~$\validator$,
determine which one was proposed first according to Lemma~\ref{lem-lisk-bft}~\ref{lem-lisk-bft-property-d}, and then check if the inequalities given in Lemma~\ref{lem-lisk-bft}~
\ref{lem-lisk-bft-property-b} hold.
\end{proof}

\section{Conclusion}

We introduced a general consensus framework for blockchains which allows to prove safety and liveness depending on the number of Byzantine validators. As a concrete application of the consensus framework we defined the \lbft protocol, which allows to obtain strong guarantees with respect to safety and liveness while requiring no additional messages for the consensus as the necessary information is efficiently
encoded in the blocks.

Another interesting possible application of the general consensus framework could be to combine it together with a
suitable aggregate signature scheme. For instance, validators could send a $\approve(h(B_l)-1,\cdot,B_l,\validator)$ message
for the block they believe to be the correct block at height~$l$. Using an aggregate signature scheme, these messages could be efficiently aggregated in the network and only added to any child block of~$B_l$ by the next validator. Applying the properties of the general consensus framework analogous to the case of the \lbft protocol, it is possible to show that a block can be finalized already after two additional blocks are added as descendants.

\printbibliography

\end{document}